\documentclass{article}
\usepackage{spconf,amsmath,graphicx}

\usepackage{amssymb}
\usepackage{amsfonts}
\usepackage{url}
\usepackage{multicol}
\usepackage{multirow}
\usepackage{hyperref}
\usepackage[capitalize]{cleveref}
\usepackage[font=small,skip=0pt]{caption}
\usepackage{subfig}
\usepackage{xcolor}
\usepackage{amsthm}         
\usepackage{xfrac}
\usepackage{subfig}
\usepackage{cite}
\usepackage{xpatch,xcolor}
\usepackage{url}

\ninept

\hypersetup{
    colorlinks=true,
    linkcolor=purple,
    filecolor=magenta,      
    urlcolor=purple,
}

\newtheorem{thm}{Theorem}

\crefname{thm}{theorem}{theorems}
\Crefname{thm}{Theorem}{Theorems}

\title{understanding probe behaviors \\ through variational bounds of mutual information}

\name{
    Kwanghee Choi\thanks{This work used the Bridges2 system at PSC and Delta system at NCSA through allocation CIS210014 from the Advanced Cyberinfrastructure Coordination Ecosystem: Services \& Support (ACCESS) program, which is supported by National Science Foundation grants \#2138259, \#2138286, \#2138307, \#2137603, and \#2138296.},
    Jee-weon Jung,
    Shinji Watanabe
}

\address{
    Carnegie Mellon University, USA
}

\begin{document}
\newcommand{\fix}[1]{\textcolor{red}{#1}}

\setlength{\abovedisplayskip}{4pt}
\setlength{\belowdisplayskip}{4pt}

\maketitle
\begin{abstract}
With the success of self-supervised representations, researchers seek a better understanding of the information encapsulated within a representation.
Among various interpretability methods, we focus on classification-based linear probing.
We aim to foster a solid understanding and provide guidelines for linear probing by constructing a novel mathematical framework leveraging information theory.
First, we connect probing with the variational bounds of mutual information (MI) to relax the probe design, equating linear probing with fine-tuning.
Then, we investigate empirical behaviors and practices of probing through our mathematical framework.
We analyze the layer-wise performance curve being convex, which seemingly violates the data processing inequality.
However, we show that the intermediate representations can have the biggest MI estimate because of the tradeoff between better separability and decreasing MI.
We further suggest that the margin of linearly separable representations can be a criterion for measuring the ``goodness of representation.''
We also compare accuracy with MI as the measuring criteria.
Finally, we empirically validate our claims by observing the self-supervised speech models on retaining word and phoneme information.
\end{abstract}


\vspace{-0.5em}
\section{Introduction} \label{sec:intro}
\vspace{-0.5em}
Despite recent advances in deep learning, each intermediate representation remains elusive due to its black-box nature.
Linear probing is a tool that enables us to observe what information each representation contains \cite{alain2016understanding,belinkov2018internal}.
Typically, a task is designed to verify whether the representation contains the knowledge of a specific interest.
Solving the task well implies that the knowledge of interest is incorporated inside the representation. 
Classification is often adopted to design such a task.
To observe the intermediate representations, the model is frozen while a learnable probe is attached to the layer to be investigated.
The probe is trained with the aforementioned task, and its test set accuracy indicates the amount of knowledge retained within the model. 
The probe usually has a simple structure, such as a single linear layer, hence called \textit{linear probing} \cite{alain2016understanding}.

However, it is known that interpreting the layer-wise accuracy performance can be tricky.
Accuracy is influenced by multiple factors, such as task or probe design.
For example, one can use the multi-layer perceptron (MLP) instead of employing a single linear layer for the probe \cite{hewitt2019designing,pimentel2020information}.
Hence, the difficulty of analyzing the accuracy ends up in vague expressions in the literature, such as goodness, quality, extractability, readability, or usability \cite{pimentel2020information,belinkov2022probing}.

To support its theoretical limitations, many advances have been made, especially utilizing information theory.
Belinkov et al. \cite{belinkov2018internal} points out that optimizing the probe using the cross-entropy loss indicates that the probe is learning to maximize the mutual information (MI) between the intermediate representations and the labels for the target task.
Pimentel et al. \cite{pimentel2020information} further incorporates probing into the information-theoretic framework, while Voita et al. \cite{voita2020information} utilizes the idea of the minimum description length.

Also, there have been debates on the implementation of probing.
Hewitt et al. \cite{hewitt2019designing} claims that adopting a large probe may end up memorizing the task, failing to give meaningful measurements.
However, others argue that the probe capacity does not affect its theoretical aspects \cite{pimentel2020information,kunz2020classifier}, similar to our viewpoint in \Cref{ssec:linprobe_eq_finetune}.
Furthermore, although the probe accuracy is known to highly correlate with how ``good'' the representation is, there is no formal definition of the ``goodness'' \cite{pimentel2020information,belinkov2022probing}, which we attempt to define in \Cref{ssec:separability}.

\begin{figure}[t] 
\centering
\includegraphics[width=1.0\columnwidth]{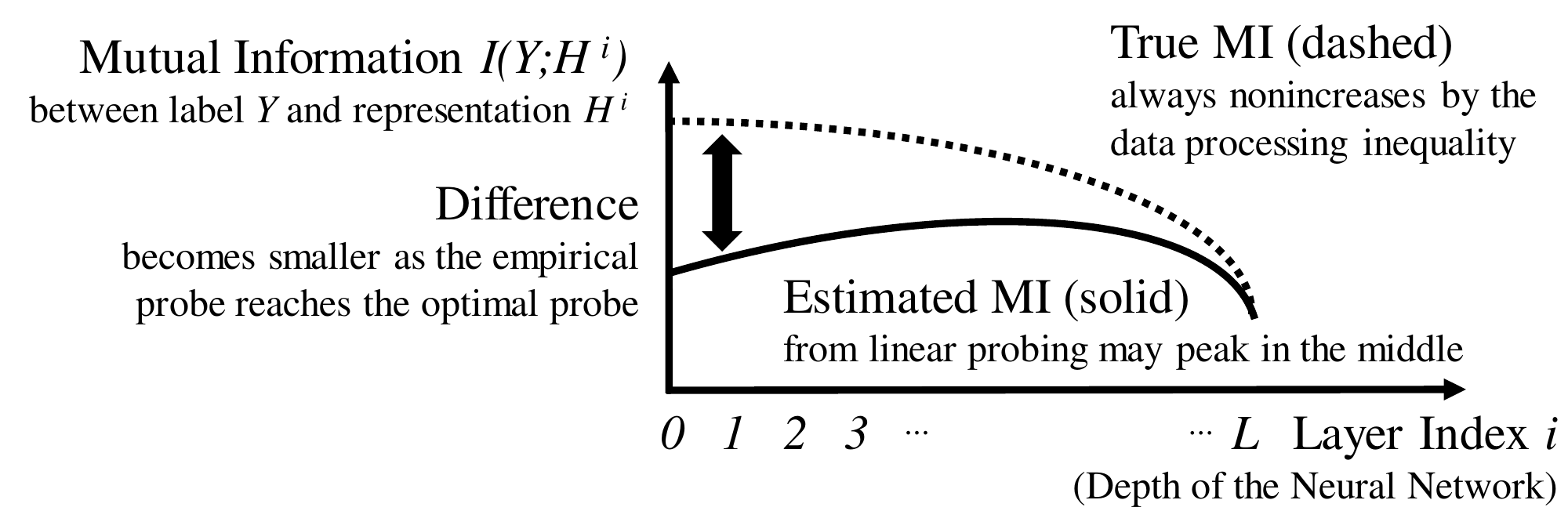}
\caption{
Mutual information $I(Y; H^i)$ between the label $Y$ and the intermediate representation $H^i$ from the $i$-th layer.
Due to the limited probe capacity of linear probing, its MI estimate fails to reach the true MI, often having a peak in the middle.
The MI estimate gets influenced by diverse factors, such as the probe capacity and the linear separability of the representations.
Details in \Cref{ssec:curved_perf}.
}
\vspace{-2em}
\label{figs:diagram}
\end{figure}
Hence, we aim to clarify the ambiguities of implementing and interpreting the results of probing.
Based on the theoretical background in \Cref{sec:theory}, we view probing through the lens of variational bounds of MI \cite{belghazi2018mutual,choi2022combating,poole2019variational} (\Cref{ssec:lin_probe_maximizes_mi_bound}).
It enables us to relax the probe's design, hence equating both linear probing and fine-tuning to maximize the lower bound of MI (\Cref{ssec:linprobe_eq_finetune}).
Based on it, we show that MI is not the only factor determining probing by discussing the curved layer-wise MI estimates (\Cref{ssec:curved_perf}). 
Further, we aim to theoretically define the ``goodness of representation'' by focusing on the margin of linear separability between intermediate features (\Cref{ssec:separability}).
Finally, we show that using the cross-entropy loss is theoretically more sound than accuracy, the latter being a loose bound for the former (\Cref{ssec:accuracy_bounded}).
We further show the experimental evidence in \Cref{sec:exp} by analyzing the self-supervised large-scale speech model and its linguistic properties.
To summarize, our contributions are threefold: (1) we propose to understand probing based on variational bounds of MI; (2) the framework enables us to understand various behaviors of probing, such as the impact of probe capacity or linear separability; and (3) we show the empirical effectiveness of our framework via probing speech models.

\section{Theoretical Background}\label{sec:theory}
\noindent \textbf{Mutual Information.} $I(V; W)$ quantifies the common information between two random variables $V$ and $W$, defined as:
\begin{equation}
    I(V; W) = D_\text{KL}(\mathbb{P}_{VW} || \mathbb{P}_V \otimes \mathbb{P}_W),
\label{eq:mi_def}
\end{equation}
where $D_\text{KL}$ is the Kullback–Leibler divergence, $\mathbb{P}_V$ and $\mathbb{P}_W$ are the marginal distributions of $V$ and $W$, $\otimes$ is the outer product, and $\mathbb{P}_{VW}$ is the joint distribution between $V$ and $W$.
If the two random variables are independent, then the joint distribution $P(v, w)$ and the product of marginals $P(v)P(w)$ becomes the same, making the divergence zero.

Entropy $H(*)$ of the random variable is the total amount of information it has.
The relationship with MI can be described as:
\begin{equation}\label{eq:entropy_upper_bounds_mi}
    0 \leq I(V; W) \leq \max (H(V), H(W)).
\end{equation}
where entropy acts as the upper bound of MI.

\noindent \textbf{Information Path.}
Consider three random variables that construct a Markov chain $U \to V \to W$.
Then, the Data Processing Inequality (DPI) \cite{cover1991elements} enables the following:
\begin{equation}
    \label{eq:three_MI}
    H(U) = I(U; U) \geq I(U; V) \geq I(U; W),
\end{equation}
where the inequality demonstrates the information of $U$ being dissipated through $V$ and $W$ \cite{shwartz2017opening,alain2016understanding}.
Let's extend the discussion to the neural network (NN) $f^{1 \ldots L}$ constructed with $L$ consecutive layers:
\begin{equation}\label{eq:layerwise_nn}
    f^{1 \ldots L} = f^{1} \circ f^{2} \circ \cdots \circ f^{L} = f^{1 \ldots i} \circ f^{i + 1 \ldots L},
\end{equation}
where the integer $1 \leq i \leq L-1$, $\circ$ is the function composition operator, and $f^1, f^2, \cdots, f^L$ denotes the 1$^{st}$, 2$^{nd}$, $\cdots$, $L$$^{th}$ layer, respectively.
Then, we can write the information path \cite{shwartz2017opening} with the intermediate features by generalizing \cref{eq:three_MI}:
\begin{align} \begin{split}\label{eq:dpi_nn}
    H(X) \geq I(X; f^{1}(X)) \geq & I(X; f^{1 \ldots 2}(X)) \\
    \geq & \cdots \geq I(X; f^{1 \ldots L}(X)),
\end{split} \end{align}
where $X$ is the input random variable and its innate information being dissipated throughout the NN.

\noindent \textbf{Donsker-Varadhan (DV) Representation} \cite{donsker1983asymptotic} is the variational dual representation of the KL divergence between two distributions:
\begin{equation}
    D_\text{KL} (\mathbb{Q}_1 || \mathbb{Q}_2) = \sup_{T: \Omega \to \mathbb{R}} \mathbb{E}_{\mathbb{Q}_1}[T] - \log \mathbb{E}_{\mathbb{Q}_2}[e^T].
\label{eq:dv_repr}
\end{equation}
where $\mathbb{Q}_1$ and $\mathbb{Q}_2$ are arbitrary two distributions, the supremum is taken over all the possible functions $T$ on a compact domain $\Omega$.
We can understand the above as finding the optimal function $T^*$ that maximizes the above equation, where $T$ receives an element $q_1 \in \mathbb{Q}_1$ and $q_2 \in \mathbb{Q}_2$ as an input and returns a single real-valued number.
If we switch $\mathbb{Q}_1$ and $\mathbb{Q}_2$ with $\mathbb{P}_{VW}$ and $\mathbb{P}_V \otimes \mathbb{P}_W$, in \cref{eq:mi_def}, DV representation becomes a variational representation of MI \cite{belghazi2018mutual,choi2022combating,poole2019variational}.

\noindent \textbf{Neural Network-Based MI Estimation.}\label{subsec:mine}
We can model the function $T$ of \cref{eq:dv_repr} with a NN $T_\theta: V \times W \to \mathbb{R}$ with parameters $\theta$, which becomes the MI Neural Estimator (MINE) loss $\mathcal{L}_\text{MINE}$ \cite{belghazi2018mutual}:
\begin{equation}\label{eq:mine}
    \mathcal{L}_\text{MINE} (T_\theta)= \mathbb{E}_{\mathbb{P}_{VW}}[T_\theta] - \log \mathbb{E}_{\mathbb{P}_V \otimes \mathbb{P}_W}[e^{T_\theta}],
\end{equation}
where we train the NN $T_\theta$ by maximizing the loss $\mathcal{L}_\text{MINE}$.
The converged loss value becomes the MI estimate $\hat{I}_\theta (V; W)$.

Also, it is known that the InfoNCE loss $\mathcal{L}_\text{InfoNCE}$, another variant of the \cref{eq:dv_repr}, is equivalent to $\mathcal{L}_\text{MINE}$ up to a constant \cite{oord2018representation}:
\begin{equation}\label{eq:infonce}
    \mathcal{L}_\text{InfoNCE} (T_{\theta})= \frac{1}{B} \sum_{i=1}^{B} \log \frac{e^{T_\theta (v_i, w_i)}}{\frac{1}{B} \sum_{j=1}^{B} e^{T_\theta (v_i, w_j)}},
\end{equation}
where $(v_i, w_i) \sim \mathbb{P}_{VW}$ is being sampled $B \in \mathbb{N}^+$ times, i.e., batch size $B$.
Interestingly, using the softmax with the cross-entropy loss during classification training becomes identical to $\mathcal{L}_\text{InfoNCE}$ \cite{choi2022combating}.


\section{Our Information-theoretic Framework}\label{sec:method}
\vspace{-0.6em}
\subsection{Linear probing maximizes the MI estimation bound}\label{ssec:lin_probe_maximizes_mi_bound}
\vspace{-0.6em}
We start with a $C$-way classification task with the input $X$ and label $Y$.
Let us assume that we can fully obtain the label information from the input \cite{shwartz2017opening}.
Then, the total information amount of $Y$ inside $X$ becomes $I(X; Y) = H(Y)$.
Now, we consider the MI between the label $Y$ and the $i^{th}$ intermediate hidden feature $H^i$ of the NN, $I(Y; f^{1 \ldots i}(X)) = I(Y; H^i)$.
We estimate the MI using $\mathcal{L}_\text{MINE}$ in \cref{eq:mine}, which is the loss form of the DV representation in \cref{eq:dv_repr}:
\begin{equation}
    I(Y; H^i) =  \mathcal{L}_\text{MINE} (T^*) \geq \hat{I}_\theta(Y; H^i) = \mathcal{L}_\text{MINE} (T_\theta), \label{eq:binary_mine}
\end{equation}
where $T^*$ is the optimal function, $T_\theta$ is the NN approximation of $T^*$ obtained by maximizing the loss $\mathcal{L}_\text{MINE}$, and $\hat{I}_\theta$ is the MI estimate produced by the NN $T_\theta$.

Note that the above formulation does not limit the structure of $T_\theta$ if it can successfully converge to the optimal function $T^*$.
Hence, following Theorem 7 in \cite{choi2022combating}, we design the linear probe $T_\theta$ as follows:
\begin{equation} \label{eq:probe_linear_design}
    T_\theta (y, h_i) = \texttt{onehot}(y) \cdot \texttt{FC}(h_i),
\end{equation}
where $\texttt{onehot}(y)$ is the $N$-dimensional one-hot label $y \in Y$, $\cdot$ is the dot product, and $\texttt{FC}$ is the learnable linear layer with the input of representation $h_i \in H^i$ and the output dimension of $N$.
If we swap $\mathcal{L}_\text{MINE}$ in \cref{eq:binary_mine} with $\mathcal{L}_\text{InfoNCE}$ in \cref{eq:infonce}, training $T_\theta$ becomes equivalent to linear probing with the $C$-way classification task.

Also, the softmax function directly estimates $P(Y|H^i)$, so that we can easily yield the MI estimate from the averaged log loss \cite{belinkov2018internal}:
\begin{equation}\label{eq:log_loss_is_mi}
    \mathbb{E}_{\mathbb{P}_{YH^i}}[-\log P(Y|H^i)] = H(Y|H^i) = H(Y) - I(Y; H^i),
\end{equation}
where $H(Y)$ is tractable when the label distribution is known.

\vspace{-0.6em}
\subsection{Linear probing is equivalent to fine-tuning}\label{ssec:linprobe_eq_finetune}
\vspace{-0.6em}
By expanding the above idea, if one changes the design of $T_\theta$ as:
\begin{equation} \label{eq:probe_finetune_design}
    T_\theta (y, h_i) = \texttt{onehot}(y) \cdot \texttt{FC}(f^{i+1 \ldots L}(h_i)),
\end{equation}
regarding the rest of the network $f^{i+1 \ldots L}$ as the probe, the formulation becomes the same with fine-tuning the later layers $f^{i+1}$, $f^{i+2}$, $\ldots$ $f^L$.
The probe capacity shrinks for the later layers $i \rightarrow L$, being the same with the linear probing for the penultimate layer $i=L$.
Nevertheless, linear probing and fine-tuning is equivalent, as variational bounds of MI do not enforce any specific designs of $T_\theta$, such as \cref{eq:probe_linear_design,eq:probe_finetune_design}.
This theoretical remark aligns well with the existing literature on probing, where they also claim that the probe capacity does not influence what probe aims to measure \cite{pimentel2020information,kunz2020classifier}.

However, MI estimation quality depends on the probe capacity that is dictated by specific design choices, such as \cref{eq:probe_linear_design,,eq:probe_finetune_design}: bigger the capacity, likely to yield better MI estimates.
The $\mathcal{V}$-information framework \cite{xu2019theory,ethayarajh2022understanding} injects the structural limitation of the probe by modifying the definition of MI.
In contrast, our framework maintains the original definition of MI while acknowledging that each MI estimate can be inaccurate.
We demonstrate the empirical effect of the probe capacity in \Cref{ssec:linear_vs_finetune,ssec:network_capacity}.

\vspace{-0.6em}
\subsection{Curved layer-wise MI estimate of linear probing} \label{ssec:curved_perf}
\vspace{-0.6em}
Based on the viewpoint of $T_\theta$ design and its probe capacity in \Cref{ssec:linprobe_eq_finetune}, we observe the linear probing behaviors in \cref{figs:diagram,,fig:linear_vs_finetune_word}.
As shown in \cref{eq:dpi_nn}, the true MI $I(Y; H^i)$ is always nonincreasing as layer index $i$ increases due to DPI (\cref{eq:three_MI}).
However, the estimated MI $\hat{I}_\theta(Y; H^i)$ may increase as the representation becomes linearly separable even though the true MI $I(Y; H^i)$ decreases.
One should regard this behavior as lost information not being somehow restored; rather, it becomes more palpable towards the linear probe.
We further experimentally show that larger probe capacity curtails the curved layer-wise MI estimate in \Cref{ssec:linear_vs_finetune}.


\vspace{-0.5em}
\subsection{Margin as the measure for ``goodness''}\label{ssec:separability}
\vspace{-0.5em}
Even though the curved layer-wise estimated MI in \Cref{ssec:curved_perf} comes from the structural limitation of a linear probe, many often choose the best-performing layer due to its empirical effectiveness \cite{pasad2023comparative}.
Existing literature often denotes the layer of maximum accuracy to be ``good'' or the innate information being ``easy to extract,'' without formal definition \cite{hewitt2019designing,pimentel2020information}.
To clarify the definition of ``goodness,'' we focus on the margin of linearly separable representations.
In detail, we show that bigger margins yield a more accurate MI estimate for the linear probe, implying the ``goodness.''

\vspace{-0.5em}
\begin{thm}\label{thm:mi_vs_margin}
    Assume a balanced binary classification task with class $0$ and $1$.
    Then, for the linearly separable data with the non-negative margin $d$, MI estimation error of the linear probe is bounded by:
    \begin{equation}
        |I(X; Y) - \hat{I}_\theta (X; Y)| < e^{-d}.
    \end{equation}
\end{thm}
\begin{proof}
    Let the decision boundary functions for class $0$ and $1$ to be $f(x) = w\cdot x + b$ such that $f(x) > 0 $ for all $x \sim \mathbb{P}_{X|Y=0}$ and $f(x) < -d $ for all $ x \sim \mathbb{P}_{X|Y=1}$.
    Then, we can design the weights of the \texttt{FC} layer inside the linear probe $T_\theta$ in \cref{eq:probe_linear_design} as:
    \begin{equation}\label{eq:probe_separable}
        T_\theta (x, y) = \texttt{onehot}(y) \cdot [w \cdot x + b; -w \cdot x -b -d],
    \end{equation}
    such that $T_\theta (x, y) > 0$ for all $(x, y) \sim \mathbb{P}_{XY}$ and $T_\theta (x, y) < -d$ for all $(x, y) \sim \mathbb{P}_{X} \otimes \mathbb{P}_{Y} \setminus \mathbb{P}_{XY}$.
    For simplicity, we will write the two distributions $\mathbb{P}_{XY}$ and $\mathbb{P}_{X} \otimes \mathbb{P}_{Y} \setminus \mathbb{P}_{XY}$ as mutually exclusive $\mathbb{Q}$ and $\Tilde{\mathbb{Q}}$, respectively.
    Then, as we assume balanced binary classification with perfect linear separability, $I(X; Y) = H(Y) = \log 2$.

    We use the $\mathcal{L}_\text{MINE}$ of \cref{eq:mine,eq:binary_mine} to estimate MI:
    \begin{align} \begin{split}
    \label{eq:separability_proof_intermediate}
        \hat{I}_\theta(X, &Y) = \mathbb{E}_{\mathbb{P}_{XY}}[T_\theta] - \log \mathbb{E}_{\mathbb{P}_X \otimes \mathbb{P}_Y}[e^{T_\theta}] \\
        &= \mathbb{E}_{\mathbb{Q}}[T_\theta] - \log \mathbb{E}_{\mathbb{Q} \cup \Tilde{\mathbb{Q}}}[e^{T_\theta}] \quad (\because \mathbb{Q} \cap \Tilde{\mathbb{Q}} = \emptyset. )\\
        &= \mathbb{E}_{\mathbb{Q}}[T_\theta] - \log \left( P(\mathbb{Q})\mathbb{E}_{\mathbb{Q}}[e^{T_\theta}] + P(\Tilde{\mathbb{Q}}) \mathbb{E}_{\Tilde{\mathbb{Q}}}[e^{T_\theta}] \right) \\
        &\quad (\because \text{By the Total Expectation Theorem.} )\\
        &= \mathbb{E}_{\mathbb{Q}}[T_\theta] - \log \left(\sfrac{1}{2}\mathbb{E}_{\mathbb{Q}}[e^{T_\theta}]  + \sfrac{1}{2}\mathbb{E}_{\Tilde{\mathbb{Q}}}[e^{T_\theta}] \right) \\
        &\quad (\because \text{The task is balanced binary classification.} )\\
        &= I(X; Y) + \mathbb{E}_{\mathbb{Q}}[T_\theta] - \log (\mathbb{E}_{\mathbb{Q}}[e^{T_\theta}] + \mathbb{E}_{\Tilde{\mathbb{Q}}}[e^{T_\theta}]). \\
        &\quad (\because I(X; Y) = \log 2.)
    \end{split} \end{align}

    \noindent
    Starting from \cref{eq:separability_proof_intermediate}, we denote $|I(X, Y) - \hat{I}(X, Y)| = \Delta$.
    As the exponential function is convex, $\log (1+x) < x$ for $x>0$, and the probe design in \cref{eq:probe_separable} bounds $T_\theta$ with $0$ for $\mathbb{Q}$ and $-d$ for $\Tilde{\mathbb{Q}}$,
    \begin{align} \begin{split}
    \Delta &= \log \left( (\mathbb{E}_{\mathbb{Q}}[e^{T_\theta}] + \mathbb{E}_{\Tilde{\mathbb{Q}}}[e^{T_\theta}]) / e^{\mathbb{E}_{\mathbb{Q}}[T_\theta]} \right ) \\
    &\leq \log (1+ \mathbb{E}_{\Tilde{\mathbb{Q}}}[e^{T_\theta}] / e^{\mathbb{E}_{\mathbb{Q}}[T_\theta]}) \leq \mathbb{E}_{\Tilde{\mathbb{Q}}}[e^{T_\theta}] / e^{\mathbb{E}_{\mathbb{Q}}[T_\theta]} \leq e^{-d} / e^{0}.\qedhere
    \end{split} \end{align}

\end{proof}

\vspace{-1.5em}
\subsection{Comparing accuracy with MI for measuring information}\label{ssec:accuracy_bounded}
\vspace{-0.5em}
\Cref{eq:log_loss_is_mi} indicates that we can easily estimate MI by averaging the classification loss.
However, existing literature often uses accuracy to quantify the information amount \cite{alain2016understanding,belinkov2022probing}, often being empirically useful \cite{pasad2023comparative}.
We suspect this is because accuracy behaves as a good surrogate for MI, even with quite a loose bound in theory.
\vspace{-0.5em}
\begin{thm}\label{thm:mi_vs_acc}
    Assume a $C$-way classification task where the prediction for the true label $y$ of input $x$ is bounded by $0 < \epsilon \leq \hat{P}_\theta(y|x) \leq 1$.
    Given the estimated accuracy $\hat{a}_\theta$, estimated MI $\hat{I}_\theta$ is bounded by:
    \begin{align} \begin{split} \label{eq:mi_acc_bound}
        H(Y) - \hat{a}_\theta\log C + (1-\hat{a}_\theta) \log \epsilon <& \\
         \hat{I}_\theta(X, Y) < H(Y)  - (1&-\hat{a}_\theta) \log 2.
    \end{split} \end{align}
\end{thm}
\begin{proof}
    We start with estimating \cref{eq:log_loss_is_mi} using the NN probe $T_\theta$:
    \begin{align} \begin{split} \label{eq:cond_entr}
        \hat{I}_\theta(X, Y) &= H(Y) + \mathbb{E}_{(x, y) \sim \mathbb{P}_{XY}}[\log \hat{P}_\theta (y|x)],
    \end{split} \end{align}
    where $H(Y)$ is tractable for the classification task when the label distribution $P(Y)$ is known, as mentioned in \Cref{ssec:lin_probe_maximizes_mi_bound}.
    Hence, if we obtain the bound of prediction $\hat{P}_\theta$ with respect to whether it was correct or not, we obtain the bound of $\hat{I}_\theta(X, Y)$:
    \begin{equation} \label{eq:prob_bound}
        \begin{cases}
            \frac{1}{C} < \hat{P}_\theta(y|x) \leq 1 &(\hat{y} = y) \\
            \epsilon \leq \hat{P}_\theta(y|x) < \frac{1}{2} &(\hat{y} \neq y)
        \end{cases}
    \end{equation}
    First, by assumption, the probability estimate is bounded by $\epsilon \leq \hat{P}_\theta(y|x) \leq 1$.
    Also, for the correct case ($\hat{y} = y$), the minimum probability estimate will be the case where probability is evenly distributed across all class predictions, with the correct class having a slightly higher value.
    On the other hand, the maximum probability estimate for the wrong case ($\hat{y} \neq y$) will be the case where the probability is distributed between top-2 class predictions.
    By the definition of accuracy, the second term of \cref{eq:cond_entr} can be split:
    \begin{align} \begin{split} \label{eq:cond_entr_accuracy}
        \mathbb{E}_{\mathbb{P}_{XY}}[\log \hat{P}_\theta (y|x)&] =
        \hat{a}_\theta \cdot \mathbb{E}_{\mathbb{P}_{XY}}[\log \hat{P}_\theta (y|x) | \hat{y} = y] \\
        &+ (1-\hat{a}_\theta) \cdot \mathbb{E}_{\mathbb{P}_{XY}}[\log \hat{P}_\theta (y|x) | \hat{y} \neq y].
    \end{split} \end{align}
    We plug in the above bounds of \cref{eq:prob_bound} into \cref{eq:cond_entr_accuracy} to obtain:
    \begin{align} \begin{split}
        \label{eq:mi_second_term_bound}
        \hat{a}_\theta \log \sfrac{1}{C} + (1-\hat{a}_\theta) \log \epsilon &< \mathbb{E}_{\mathbb{P}_{XY}}[\log \hat{P}_\theta (y|x)] \\
        &< \hat{a}_\theta \log 1 + (1-\hat{a}_\theta) \log \sfrac{1}{2}. 
    \end{split} \end{align}
    By substituting \cref{eq:mi_second_term_bound} into \cref{eq:cond_entr}, we obtain \cref{eq:mi_acc_bound}. \qedhere
\end{proof}
Based on \cref{eq:prob_bound}, we can construct an example where MI varies largely with accuracy, even though the example would be unlikely in real-world cases.
By having additional assumptions, we can obtain tighter bounds, which are thoroughly explored in the recent $\mathcal{H}$-consistency literature \cite{mao2023cross}.
We further show the empirical behaviors comparing MI and accuracy in \Cref{ssec:linear_vs_finetune}.

\begin{figure*}[t!]
  \centering
  \subfloat[Word classification. \label{fig:linear_vs_finetune_word}]{
    \includegraphics[height=0.2\textwidth]{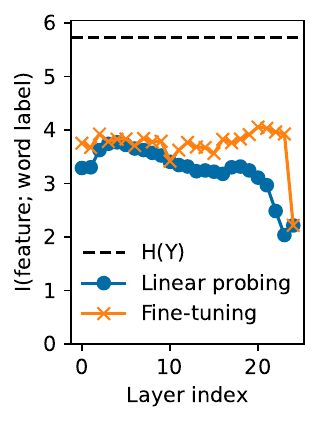}
    \includegraphics[height=0.2\textwidth]{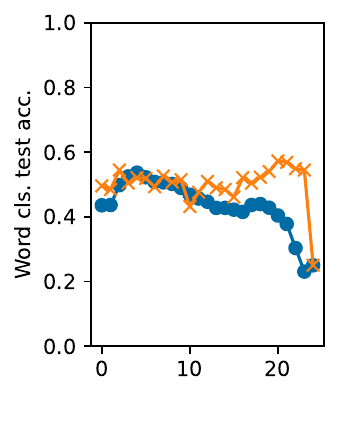}
  }
  \subfloat[Phoneme classification. \label{fig:linear_vs_finetune_phoneme}]{
    \includegraphics[height=0.2\textwidth]{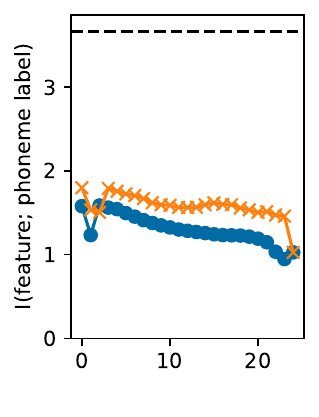}
    \includegraphics[height=0.2\textwidth]{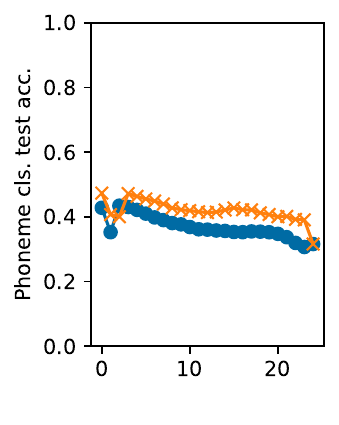}
  }
  \subfloat[Comparing different probe capacities. \label{fig:mlp_probing}]{
    \includegraphics[height=0.2\textwidth]{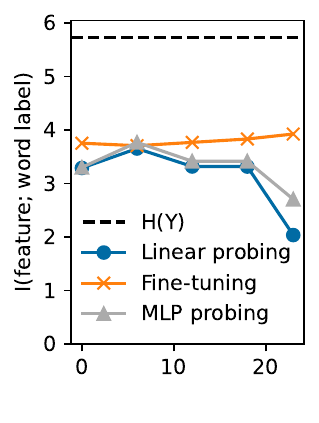}
    \includegraphics[height=0.2\textwidth]{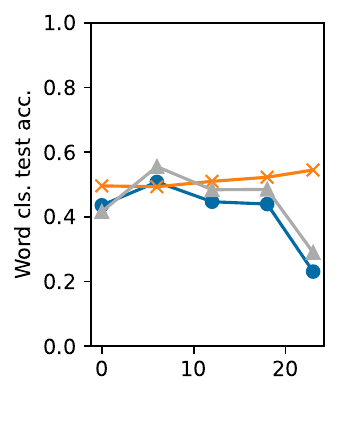}
  }
  \caption{
  Various probing experiments.
  We compare different tasks (word and phoneme classification), probe capacities (linear probing, fine-tuning, and MLP probing), and ways of measuring information (MI and accuracy).
  Upper bound for MI and accuracy is $H(Y)$ (by \cref{eq:entropy_upper_bounds_mi}) and 100\%, meaning that the full information on the label $Y$ is retrieved from the layer-wise representation $H^i$.
  }
\vspace{-1.5em}
\end{figure*}
\vspace{-0.5em}
\section{Experiment} \label{sec:exp}
\vspace{-0.5em}
\subsection{Experiment settings} \label{ssec:db}
\vspace{-0.5em}
To empirically validate our claims in \Cref{sec:method}, we probe on self-supervised large-scale speech models, which have recently started drawing attention \cite{pasad2021layer,pasad2023comparative,pasad2023self}.
We choose the word and phoneme classification task, following the setting of \cite{pasad2021layer,pasad2023comparative}.
We use the CommonPhone dataset \cite{klumpp2022common}, which provides the word and phoneme alignments.
It contains gender-balanced speech in English, German, Spanish, French, Italian, and Russian, summing up to 116.5 hours.
We removed samples longer than 2 seconds due to GPU constraints.
We further removed the class with less than 200 and 1000 samples for words and phonemes to avoid the severe long-tailed behavior of the dataset, ending up with 76 words and 99 phonemes.


Given the time-aligned words or phonemes, we derive the feature by inputting the sliced audio\footnote{\cite{pasad2021layer,pasad2023comparative} extracted the features within the full utterance.} and averaging through the temporal dimension.
To support various languages, we explore the self-supervised cross-lingual XLS-R model \cite{babu2022xls}.
We chose the smallest XLS-R model with 24 layers of Transformer encoder architecture, a total of 300M parameters.
We utilize the provided train, valid, and test set splits from the dataset and train the probes for 50 and 3 epochs for word and phoneme classification while evaluating the model every 1 and 0.1 epoch.
We use the AdamW optimizer \cite{loshchilov2019decoupled} with default settings, with a learning rate of 5e-4, and employ early stopping with the MI estimate on the validation set.
The source codes for the experiments are publicly available.\footnote{\url{https://github.com/juice500ml/information\_probing}}

\vspace{-0.5em}
\subsection{Comparing linear probing with fine-tuning} \label{ssec:linear_vs_finetune}
\vspace{-0.5em}
\Cref{ssec:linprobe_eq_finetune} emphasizes the theoretical equivalence between linear probing (\cref{eq:probe_linear_design}) and fine-tuning (\cref{eq:probe_finetune_design}).
However, due to their difference in the probe capacity, we can expect that the former will yield a more inaccurate MI estimate.
As the estimated MI $\hat{I}_\theta(Y;H^i)$ is upper-bounded by the true MI $I(Y;H^i)$ (\cref{eq:binary_mine}), worse estimates always mean lower estimates.
We can observe in \cref{fig:linear_vs_finetune_word,fig:linear_vs_finetune_phoneme} that empirical results follow the theoretical conclusion; the MI estimate of the linear probing is always lower than that of fine-tuning.
As the fine-tuning probe's capacity shrinks to match that of the linear probe as the layer index $i \rightarrow L$, the MI estimate decreases steeply in the penultimate layer.
Further, in \cref{fig:linear_vs_finetune_word}, we can observe the curved layer-wise MI estimates of linear probing described in \Cref{ssec:curved_perf}, where the middle layer showing the largest MI estimate.
Comparison between linear probing and fine-tuning further experimentally supports the claim in \Cref{ssec:curved_perf}.
More accurate MI estimates from fine-tuning indicate that the word information $I(Y; H^i)$ does not magically appear in the intermediate representations.
Rather, it is retained nonlinearly for the later representations, where the linear probe underestimates the MI due to its structural limitation.
We also compare the accuracy $\hat{a}_\theta$ and the estimated MI $\hat{I}_\theta$ in \cref{fig:linear_vs_finetune_word}.
Unlike MI, fine-tuning shows slightly worse accuracy compared to linear probing in the earlier layers.
It indicates that accuracy is a good proxy for MI while not having the exact equivalence, further supporting the conclusion of \Cref{ssec:accuracy_bounded}.
Finally, the curved layer-wise MI estimate of linear probing does not appear in phoneme classification (\cref{fig:linear_vs_finetune_phoneme}).
We suspect that this behavior is due to the different levels of task difficulty (further covered in \Cref{ssec:task_difficulty}).

\vspace{-0.5em}
\subsection{Testing different probe capacities} \label{ssec:network_capacity}
\vspace{-0.5em}
\Cref{eq:binary_mine} implies that, as the NN probe $T_\theta$ approaches the optimal function $T^*$, its MI estimate $\hat{I}_\theta(Y;H^i)$ gets closer to the true MI $I(Y;H^i)$.
\Cref{ssec:linprobe_eq_finetune} further claims that the bigger probe capacity will yield a better MI estimation.
Hence, we employ MLP with the hidden dimension of $1000$ and \texttt{ReLU} activation as the probe \cite{hewitt2019designing}, where we expect the probe capacity of the MLP probing to position between linear probing and fine-tuning.
MLP is often used in the probing literature based on the universal approximation theorem \cite{pinkus1999approximation}, \textit{i.e.,} sufficient hidden unit size with nonlinearity can approximate any function, even though it cannot guarantee the convergence of the probe $T_\theta$ to the optimal $T^*$.
We use the same word classification setting in \Cref{ssec:linear_vs_finetune}, where we probe on fewer layers due to computational constraints, layer index $i \in \{0, 6, 12, 18, 23\}$.
Experimental results of \cref{fig:mlp_probing} validate our claim in \Cref{ssec:linprobe_eq_finetune}.
The estimated MI gets smaller in the order of fine-tuning, MLP probing, and linear probing; probe capacity from big to small.

\vspace{-0.5em}
\begin{figure}[t]
  \centering
  \subfloat{\includegraphics[height=0.21\textwidth]{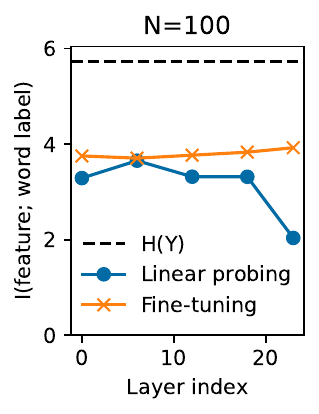}}
  \subfloat{\includegraphics[height=0.21\textwidth]{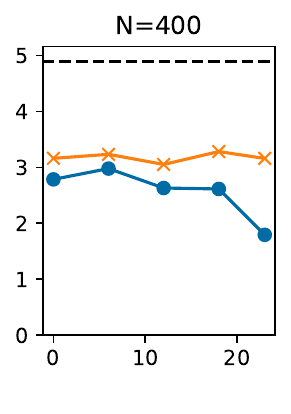}}
  \subfloat{\includegraphics[height=0.21\textwidth]{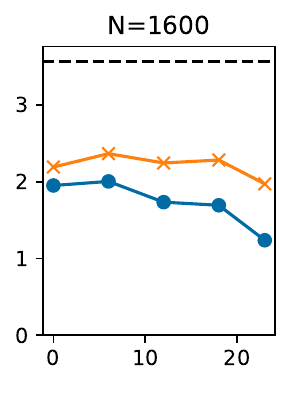}}
  \caption{
  $N$ denotes the minimum classwise sample limit for word classification.
  As $N$ increases, task gets more easier, and the curved layer-wise MI estimate of linear probing diminishes.
  }
  \vspace{-10pt}
  \label{fig:task_complexity}
\end{figure}
\vspace{-0.5em}
\subsection{Impact of task difficulty} \label{ssec:task_difficulty}
\vspace{-0.5em}
We further focus on the task-dependent behaviors of linear probing in \Cref{ssec:linear_vs_finetune}; layer-wise MI estimates being curved in word classification (\cref{fig:linear_vs_finetune_word}) while nondecreasing in phoneme classification (\cref{fig:linear_vs_finetune_phoneme}).
To systematically compare, we adopt a hyperparameter to control the probing task difficulty: minimum classwise sample limit $N$.
By changing the limit $N$ for word classification, we can directly control the difficulty.
As the sample limit $N$ decreases, the number of classes $C$ increases, and the distribution becomes long-tailed, increasing task difficulty \cite{kye2023tidal}.
\cref{fig:task_complexity} demonstrates that as task difficulty decreases, the layer-wise MI estimates change; initially being curved, it slowly changes into nonincreasing estimates.
The experimental results imply the task difference between the word and the phoneme classification, the task design influencing the linear probing results.
However, fine-tuning does not show a significant difference even with different task difficulties, showing the method's stability.
Even though $H(Y)$ greatly differs as $N$ changes, the normalized MI estimate $\hat{I}_\theta(H^i; Y)/H(Y)$ stays roughly the same, indicating the robustness against the hyperparameter $N$.

\vspace{-0.5em}
\section{Conclusion} \label{sec:concl}
\vspace{-0.5em}
We explored probing through the lens of information theory to deepen the understanding of interpreting the layer-wise MI estimates.
We first showed that linear probing and fine-tuning are equivalent in terms of variational bounds of MI.
Then, we demonstrated that the curved layer-wise MI estimate stems from the architectural constraint, while the margin of linearly separable representations plays a role where it can be used to quantify the ``goodness of representation.''
Finally, we analyze why accuracy behaves as a reasonable yet limited proxy for MI.
Our experimental results on self-supervised speech models support our theoretical observations.


\clearpage
\bibliographystyle{IEEEbib}
\bibliography{refs}

\begin{thebibliography}{10}

\bibitem{alain2016understanding}
Guillaume Alain and Yoshua Bengio,
\newblock ``Understanding intermediate layers using linear classifier probes,''
\newblock in {\em International Conference on Learning Representations,
  Workshop Track Proceedings}, 2017.

\bibitem{belinkov2018internal}
Yonatan Belinkov,
\newblock {\em On internal language representations in deep learning: An
  analysis of machine translation and speech recognition},
\newblock Ph.D. thesis, Massachusetts Institute of Technology, 2018.

\bibitem{hewitt2019designing}
John Hewitt and Percy Liang,
\newblock ``Designing and interpreting probes with control tasks,''
\newblock in {\em Conference on Empirical Methods in Natural Language
  Processing (EMNLP)}, Hong Kong, China, 2019, pp. 2733--2743, Association for
  Computational Linguistics.

\bibitem{pimentel2020information}
Tiago Pimentel, Josef Valvoda, Rowan~Hall Maudslay, Ran Zmigrod, Adina
  Williams, and Ryan Cotterell,
\newblock ``Information-theoretic probing for linguistic structure,''
\newblock in {\em Annual Meeting of the Association for Computational
  Linguistics}, Online, 2020, pp. 4609--4622, Association for Computational
  Linguistics.

\bibitem{belinkov2022probing}
Yonatan Belinkov,
\newblock ``Probing classifiers: Promises, shortcomings, and advances,''
\newblock {\em Computational Linguistics}, vol. 48, no. 1, pp. 207--219, 2022.

\bibitem{voita2020information}
Elena Voita and Ivan Titov,
\newblock ``Information-theoretic probing with minimum description length,''
\newblock in {\em Conference on Empirical Methods in Natural Language
  Processing (EMNLP)}, Online, 2020, pp. 183--196, Association for
  Computational Linguistics.

\bibitem{kunz2020classifier}
Jenny Kunz and Marco Kuhlmann,
\newblock ``Classifier probes may just learn from linear context features,''
\newblock in {\em Proceedings of the 28th International Conference on
  Computational Linguistics}, Barcelona, Spain (Online), 2020, pp. 5136--5146,
  International Committee on Computational Linguistics.

\bibitem{belghazi2018mutual}
Mohamed~Ishmael Belghazi, Aristide Baratin, Sai Rajeswar, Sherjil Ozair, Yoshua
  Bengio, R.~Devon Hjelm, and Aaron~C. Courville,
\newblock ``Mutual information neural estimation,''
\newblock in {\em International Conference on Machine Learning, {ICML}}. 2018,
  vol.~80 of {\em Proceedings of Machine Learning Research}, pp. 530--539,
  {PMLR}.

\bibitem{choi2022combating}
Kwanghee Choi and Siyeong Lee,
\newblock ``Combating the instability of mutual information-based losses via
  regularization,''
\newblock in {\em Uncertainty in Artificial Intelligence, Proceedings of the
  Thirty-Eighth Conference on Uncertainty in Artificial Intelligence, {UAI}
  2022, 1-5 August 2022, Eindhoven, The Netherlands}. 2022, vol. 180 of {\em
  Proceedings of Machine Learning Research}, pp. 411--421, {PMLR}.

\bibitem{poole2019variational}
Ben Poole, Sherjil Ozair, A{\"{a}}ron van~den Oord, Alex Alemi, and George
  Tucker,
\newblock ``On variational bounds of mutual information,''
\newblock in {\em International Conference on Machine Learning, {ICML}}. 2019,
  vol.~97 of {\em Proceedings of Machine Learning Research}, pp. 5171--5180,
  {PMLR}.

\bibitem{cover1991elements}
Thomas~M. Cover and Joy~A. Thomas,
\newblock {\em Elements of information theory},
\newblock Wiley, New York, 1991.

\bibitem{shwartz2017opening}
Ravid Shwartz-Ziv and Naftali Tishby,
\newblock ``Opening the black box of deep neural networks via information,''
\newblock {\em ArXiv preprint}, vol. abs/1703.00810, 2017.

\bibitem{donsker1983asymptotic}
Monroe~D Donsker and SR~Srinivasa Varadhan,
\newblock ``Asymptotic evaluation of certain markov process expectations for
  large time. iv,''
\newblock {\em Communications on pure and applied mathematics}, vol. 36, no. 2,
  pp. 183--212, 1983.

\bibitem{oord2018representation}
Aaron van~den Oord, Yazhe Li, and Oriol Vinyals,
\newblock ``Representation learning with contrastive predictive coding,''
\newblock {\em ArXiv preprint}, vol. abs/1807.03748, 2018.

\bibitem{xu2019theory}
Yilun Xu, Shengjia Zhao, Jiaming Song, Russell Stewart, and Stefano Ermon,
\newblock ``A theory of usable information under computational constraints,''
\newblock in {\em International Conference on Learning Representations,
  {ICLR}}. 2020, OpenReview.net.

\bibitem{ethayarajh2022understanding}
Kawin Ethayarajh, Yejin Choi, and Swabha Swayamdipta,
\newblock ``Understanding dataset difficulty with v-usable information,''
\newblock in {\em International Conference on Machine Learning, {ICML}}. PMLR,
  2022, pp. 5988--6008.

\bibitem{pasad2023comparative}
Ankita Pasad, Bowen Shi, and Karen Livescu,
\newblock ``Comparative layer-wise analysis of self-supervised speech models,''
\newblock in {\em IEEE International Conference on Acoustics, Speech and Signal
  Processing (ICASSP)}. IEEE, 2023, pp. 1--5.

\bibitem{mao2023cross}
Anqi Mao, Mehryar Mohri, and Yutao Zhong,
\newblock ``Cross-entropy loss functions: Theoretical analysis and
  applications,''
\newblock in {\em International Conference on Machine Learning, {ICML}}. 2023,
  vol. 202 of {\em Proceedings of Machine Learning Research}, pp. 23803--23828,
  {PMLR}.

\bibitem{pasad2021layer}
Ankita Pasad, Ju-Chieh Chou, and Karen Livescu,
\newblock ``Layer-wise analysis of a self-supervised speech representation
  model,''
\newblock in {\em IEEE Automatic Speech Recognition and Understanding Workshop
  (ASRU)}. IEEE, 2021, pp. 914--921.

\bibitem{pasad2023self}
Ankita Pasad, Chung-Ming Chien, Shane Settle, and Karen Livescu,
\newblock ``What do self-supervised speech models know about words?,''
\newblock {\em ArXiv preprint}, vol. abs/2307.00162, 2023.

\bibitem{klumpp2022common}
Philipp Klumpp, Tomas Arias, Paula~Andrea P{\'e}rez-Toro, Elmar Noeth, and Juan
  Orozco-Arroyave,
\newblock ``Common phone: A multilingual dataset for robust acoustic
  modelling,''
\newblock in {\em Language Resources and Evaluation Conference}, Marseille,
  France, 2022, pp. 763--768, European Language Resources Association.

\bibitem{babu2022xls}
Arun Babu, Changhan Wang, Andros Tjandra, Kushal Lakhotia, Qiantong Xu, Naman
  Goyal, Kritika Singh, Patrick von Platen, Yatharth Saraf, Juan Pino, Alexei
  Baevski, Alexis Conneau, and Michael Auli,
\newblock ``{XLS-R:} self-supervised cross-lingual speech representation
  learning at scale,''
\newblock in {\em Annual Conference of the International Speech Communication
  Association (Interspeech)}. 2022, pp. 2278--2282, {ISCA}.

\bibitem{loshchilov2019decoupled}
Ilya Loshchilov and Frank Hutter,
\newblock ``Decoupled weight decay regularization,''
\newblock in {\em International Conference on Learning Representations,
  {ICLR}}. 2019, OpenReview.net.

\bibitem{pinkus1999approximation}
Allan Pinkus,
\newblock ``Approximation theory of the mlp model in neural networks,''
\newblock {\em Acta numerica}, vol. 8, pp. 143--195, 1999.

\bibitem{kye2023tidal}
Seong~Min Kye, Kwanghee Choi, and Buru Chang,
\newblock ``Tidal: Learning training dynamics for active learning,''
\newblock in {\em International Conference on Computer Vision, {ICCV}}, 2023.

\end{thebibliography}

\end{document}